\documentclass[reqno, eucal]{amsart}

\usepackage[mathscr]{eucal} 
\usepackage[dvips]{graphicx}
\usepackage[dvips]{color}
\usepackage{amsmath,amsfonts,amssymb,amsthm,amscd}
\usepackage{epic}
\usepackage{eepic}
\usepackage{longtable}
\usepackage{array}
\usepackage{float}
\usepackage{here}

\usepackage[dvipdfmx]{hyperref}

\newtheorem{theorem}{Theorem}
\newtheorem{lemma}[theorem]{Lemma}

\newtheorem{corollary}[theorem]{Corollary}

\theoremstyle{definition}

\newtheorem{example}[theorem]{Example}

\newcommand{\Z}{{\mathbb Z}}

\newcommand{\I}{{\mathrm i}}
\newcommand{\e}{{\mathrm e}}

\begin{document}

\title[RG theoretical perturbation]
{A remark on renormalization group theoretical perturbation  \\
in a class of ordinary differential equations}

\author{Atsuo Kuniba}
\address{Atsuo Kuniba: Institute of Physics, 
University of Tokyo, Komaba, Tokyo 153-8902, Japan}
\email{atsuo.s.kuniba@gmail.com}

\maketitle

\vspace{0.5cm}
\begin{center}{\bf Abstract}
\end{center}
We revisit the renormalization group (RG) theoretical perturbation theory on  
oscillator-type second-order ordinary differential equations.
For a class of potentials,  we show a simple functional relation 
among secular coefficients of the harmonics
in the naive perturbation series.
It leads to an inversion formula between bare and renormalized amplitudes 
and an elementary proof of the absence of secular terms in all orders of the RG series.
The result covers nonautonomous as well as autonomous cases and 
refines earlier studies, including the classic examples 
of Van der Pol, Mathieu, Duffing, and Rayleigh equations.

\vspace{0.4cm}

\section{Introduction}

In second-order ordinary differential equations such as 
\begin{align}
\text{Van der Pol}:& \quad\frac{d^2y}{dt^2} + y+ \varepsilon (y^2-1)\frac{dy}{dt}= 0,
\label{vdp}
\\
\text{Mathieu}: & \quad\frac{d^2y}{dt^2} + y + \varepsilon (g+ 2 \cos t) y=0,
\label{mathi}
\\
\text{Duffing}: &\quad\frac{d^2y}{dt^2} + y +\varepsilon(\frac{dy}{dt} + g y^3)=0,
\label{duff} 
\\
\text{Rayleigh}: &\quad
\frac{d^2y}{dt^2}  + y + \varepsilon\frac{dy}{dt}\Bigl(\frac{1}{3}\bigl(\frac{dy}{dt}\bigr)^2-1\Bigr) = 0,
\label{ray}
\end{align}
naive perturbation around $\varepsilon=0$ leads to a series of the form
\begin{align}\label{yyb}
y = \sum_{k \ge 0} \sum_{n \in \mathbb{Z}}(\text{polynomial in $t$}) \varepsilon^k\mathrm{e}^{n \I t}.
\end{align}
The polynomial nature of the coefficients is called {\em secular} and invalidates 
any effective description beyond a typical time scale of 
$\mathcal{O}(\varepsilon^{-1})$.
The renormalization group (RG) theoretical approach 
is a successful example of singular perturbation \cite{BO, H, KT}
which circumvents the difficulty and offers an effective 
resummation of the divergent series. 
The basic strategy is to absorb the secular $t$-dependence of 
Eq.~(\ref{yyb}) into
renormalized amplitudes and describe the slow dynamics of the latter by 
the so-called RG equation.
The method has a rich background and perspectives;
see, for example, \cite{CGO,K,NO,O,T} and the references therein.

This short note is an elementary and modest addition to the well-developed machinery.
We focus on equations of the form
\begin{align}
\frac{d^2y}{dt^2}  + y  = \varepsilon V,
\qquad 
V  \;\;\text{an arbitrary polynomial in $\varepsilon,  
\e^{\pm \I t}, y$ and $\frac{dy}{dt}$},
\label{feq0}
\end{align}
which are a bit specific but cover Eqs.~(\ref{vdp})--(\ref{ray}) well.
The linear and $\frac{dy}{dt}$-free case like 
the Mathieu equation, Eq.~(\ref{mathi}), may also be viewed as 
a stationary Schr\"odinger equation in one-dimensional 
periodic potentials expressible with finitely many Fourier components. 

Let $\sum_{n \in \mathbb{Z}} P_n(\varepsilon,t,A,B)\e^{n\I t}$ 
be the naive perturbation series of Eq.~(\ref{yyb})
that reduces to $A\e^{\I t} + B \e^{-\I t}$ at $\varepsilon=0$.
Here $P_n(\varepsilon,t,A,B)$ is a power series in $\varepsilon$ 
which we call the secular coefficient. 
In particular,  
$P_{\pm1}(\varepsilon,t,A,B)$ are important in that the relevant $\e^{\pm \I t}$ 
are the {\em resonant} harmonics.
Our main finding is the functional relation (Corollary \ref{co:sasaki}) 
\begin{align}\label{ppp}
P_n(\varepsilon,t,A,B) = P_n(\varepsilon,t-s, P_1(\varepsilon,s,A,B),P_{-1}(\varepsilon,s,A,B)) \quad (\forall n \in \mathbb{Z}),
\end{align}
and the resulting refinements in the conventional RG approach.
For example, absence of secular terms 
in all orders of the RG series\footnote{A general proof of this basic fact
seems lacking in the literature.} follows immediately, and 
the manifest bijection in Eq.~(\ref{sasaki}) 
between the bare and renormalized amplitudes 
makes it unnecessary to introduce the 
so-called renormalization constants and resort to
the implicit function theorem.

Our proof of Eq.~(\ref{ppp}) is elementary 
and elucidates why it works naturally for 
$V$ of the form in Eq.~(\ref{feq0}).
\footnote{The form of $V$ in Eq.~(\ref{V})
 is a sufficient but {\em not} a {\em necessary} condition
for the  approach in this paper to work. 
See the remark after the proof of Lemma \ref{le:matsuko}.}
It would be interesting to explore a geometric and/or holographic 
interpretation of it \cite{C,N}, and 
to seek a similar structure in a wider class of differential equations. 

In Sect.~\ref{sec:na} a precise definition of $P_n(\varepsilon,t,A,B)$ is given.
In Sect.~\ref{sec:ps} the main result, Eq.~(\ref{ppp}), is proved.
In Sect.~\ref{sec:ra} the 
renormalized amplitude is introduced 
and the RG equation is derived in a one-line calculation.
In Sect.~\ref{sec:ex} the classical examples of Eqs.~(\ref{vdp})--(\ref{ray}) 
are treated along the scheme of the paper. 
Similar analyses are available in many works and 
we have no intention of claiming originality, 
not to mention the basic idea and flow of the RG analysis.
The last Sect.~\ref{ss:saori} includes an exercise 
on a nonlinear and nonautonomous example.

\section{Naive perturbation}\label{sec:na}

We study the second-order ordinary differential equation for $y=y(t)$ of the form
\begin{align}
\frac{d^2y}{dt^2}  + y  = \varepsilon V\Bigl(\varepsilon, \e^{\I t}, \e^{-\I t},y,\frac{dy}{dt}\Bigr).
\label{feq}
\end{align}
Here $\varepsilon$ is a parameter with respect 
to which the perturbation series is to be constructed.
The function 
$V=V\bigl(\varepsilon, \e^{\I t}, \e^{-\I t},y,\frac{dy}{dt}\bigr)$, 
which we call the {\em potential}, 
is a polynomial in the five variables indicated.
That is, we assume that $V$ has the form 
\begin{align}\label{V}
V = \sum_{k \in \Z,\, l,m,n \in \Z_{\ge 0}} C_{k l m n} \,\varepsilon^n \e^{k\I t}y^l 
\Bigl(\frac{dy}{dt}\Bigr)^m,
\end{align} 
where $C_{k l m n}$ is a coefficient independent of $\varepsilon$ and $t$, 
and is nonzero only for finitely many quartets $(k,l,m,n)$.
We work in the generic complex domain, so 
$C_{klmn} = C_{-klmn}^\ast$ need not be imposed.
Thus for example, $V = \e^{2\I t} y^2 + (4+\varepsilon \e^{-5\I t}) y\bigl(\frac{dy}{dt} )^3+ 7 \e^{-3\I t}$ 
is covered but $V = ty$ is not.
It should also be noted that $V$ is allowed to contain 
only {\it commensurate} driving terms in that
$k \in \Z$ in Eq.~(\ref{V}).
We would like to remove the trivial linear case
$V = C_1y + C_2\frac{dy}{dt} + \sum_k D_k\e^{k\I t}$,
so the existence of nonzero $C_{k l m n}$ with
$l+m \ge 1$ and $|k|\ge  \max(2-l-m,0)$ is assumed.

In naive perturbation we will be concerned with solutions of the form
\begin{align}
y(\varepsilon,t) = \sum_{n \in \Z}  \sum_{k\in \Z_{\ge 0}}\varepsilon^k f_{n,k}(t) \e^{n \I t},
\qquad f_{n,k}(t) \;\; \text{a polynomial in $t$}.
\label{ysol}
\end{align}
This is a formal power series in $\varepsilon$ 
and is also a formal Laurent series in $\e^{\I t}$.
As we will see, our construction always leads to  
$f_{n,k}(t) = 0$ for sufficiently large $|n|$ for each fixed $k$.
Therefore, any product, say $y^5\bigl(\frac{dy}{dt}\bigr)^3$, makes sense 
as a formal Laurent series in $\e^{\I t}$.

Consider the formal power-series expansion 
\begin{align}
y(\varepsilon, t) = y_0(t) + \varepsilon y_1(t) + \varepsilon^2 y_2(t) + \cdots,
\label{yi}
\end{align}
which corresponds to setting 
$y_k(t)  = \sum_{n \in \Z}f_{n,k}(t)\e^{n\I t}$  in Eq.~(\ref{ysol}).
Substituting Eq.~(\ref{yi}) into Eq.~(\ref{feq}), 
we get an equation for each power of $\varepsilon$:
\begin{align}
\frac{d^2 y_0}{dt^2} + y_0 &= 0,
\label{yeq0}
\\
\frac{d^2 y_1}{dt^2} + y_1 &= V\Bigl(0,\e^{\I t}, \e^{-\I t},y_0,\frac{dy_0}{dt}\Bigr),
\\
& \vdots
\\
\frac{d^2 y_k}{dt^2} + y_k &= \Bigl[V\Bigl(\varepsilon,\e^{\I t}, \e^{-\I t},
\sum_{j=0}^{k-1}\varepsilon^jy_j,
\sum_{j=0}^{k-1}\varepsilon^j\frac{dy_j}{dt}\Bigr)\Bigr]_{\varepsilon^{k-1}}\quad (k \ge 2).
\label{yeqn}
\end{align}
Here and in what follows, 
$[ \cdots ]_{x}$ denotes the coefficient of $x$.
The general solution to Eq.~(\ref{yeq0}) is 
\begin{align}
y_0(t) = A \e^{\I t} + B \e^{-\I t},
\label{y0}
\end{align}
where $A$ and $B$ are two independent complex parameters.
Starting from this, we can successively determine 
$y_k(t)$ with $k=1,2,3\ldots$ uniquely up to the freedom 
of adding $\alpha_k \e^{\I t} + \beta_k \e^{-\I t}$ 
for arbitrary $\alpha_k, \beta_k$ 
for each $k$.
Let us choose these $\alpha_k, \beta_k$ so that 
the constant term of $f_{\pm1, k}(t)$ becomes zero for all $k \ge 1$.
That is, we demand\footnote{There is no loss of generality compared with setting
$f_{\pm 1,k}(t=t_0) = 0$ with another parameter $t_0$.
This degree of freedom is essentially incorporated into the forthcoming 
Eq.~(\ref{cdk}).
The present convention was employed implicitly in Ref. \cite{K}.}
\begin{align}
f_{\pm 1,k}(t=0) = 0 \quad (\forall k \ge 1).
\label{f1}
\end{align}
This completely fixes $\alpha_k, \beta_k$, again successively, 
for $k=1,2,3,\ldots$ 
From the construction, it is easy to see the property
$f_{n,k}(t) = 0\, (|n|\gg 1, \, k\,\text{fixed})$  mentioned after Eq.~(\ref{ysol}).

We define $Y(\varepsilon,t,A,B)$ 
to be the resulting formal solution  $y(\varepsilon,t)$ in Eq.~(\ref{yi}).
It can also depend on other parameters in the potential $V$, 
like $g$ in Eqs.~(\ref{mathi}) and (\ref{duff}).
This last class of parameters 
inherent in $V$ will be suppressed in the notation below.
By using $Y(\varepsilon,t,A,B)$, 
define the quantity $P_n(\varepsilon, t,A,B)$ to be the coefficients 
occurring in the expansion into harmonics:  
\begin{align}
Y(\varepsilon,t,A,B) = \sum_{n \in \Z}P_n(\varepsilon, t,A,B)\e^{n\I t}.
\label{yp}
\end{align}
In terms of Eq.~(\ref{ysol}), this means setting 
$P_n(\varepsilon, t,A,B) = \sum_{k \ge 0}\varepsilon^k f_{n,k}(t)$.
We call  $A,B$  the {\em bare} amplitudes and 
$P_n(\varepsilon, t,A,B)$ the {\em secular coefficient} of the harmonics $\e^{n \I t}$.
The special case  $P_{\pm 1}(\varepsilon,t,A,B)$, 
which is called the {\em resonant secular coefficient} \cite{NO},
will play a key role in what follows. 
By definition they satisfy 
\begin{align}
&P_{\pm 1}(\varepsilon,0,A,B) =\begin{pmatrix}A \\ B \end{pmatrix},
\label{pt0}
\\
&P_n(0,t,A,B) = A\delta_{n,1}+  B\delta_{n,-1}.
\label{pe0}
\end{align}
In contrast to Eq.~(\ref{pt0}), $P_n(\varepsilon, 0,A,B)$ with $n \neq \pm 1$ is 
nontrivial in general, and will take part in the renormalized expansion; 
see Eq.~(\ref{deema2}).
The range of the $n$-sum in Eq.~(\ref{yp}) 
can be some subset of $\Z$ depending on the 
potential $V$.
For instance, in the list in Eqs.~(\ref{vdp}) -- (\ref{ray}),  
the sums are actually $\sum_{n \in 2\Z+1}$ 
except for Eq.~(\ref{mathi}).
In another example, $V = y \e^{2\I t}$, 
the sum reduces to $\sum_{n \in 2 \Z_{\ge 0}+1}$.

\begin{example} \label{ex:vdp}
Consider Van der Pol equation in Eq.~(\ref{vdp}), which corresponds to taking the 
potential as $V = (1-y^2) \frac{dy}{dt}$ in Eq.~(\ref{feq}).
Then the above definition leads to $P_n(\varepsilon,t,A,B) =0$ for $n$ even, 
as mentioned.
For odd $n$ we have the expansions in Table \ref{tab:hpc}, 
where the $k$th row $k\ge 0$ and the 
$n$th column $n\ge 1$ from the top left corner 
shows the polynomial $f_{k,2n-1}(t)$ 
in Eq.~(\ref{ysol}) corresponding to $Y(\varepsilon,t,A,B)$.
We have set $C=AB$ to save space.

\vspace{0.2cm}
\begin{table}[h]
\caption{Expansions for odd $n$ of the Van der Pol equation in Eq.~(\ref{vdp}).}
\label{tab:hpc}
{\small
\begin{tabular}{c|c|c|c|c}
& $\e^{it}$ & $\e^{3it}$ & $\e^{5it}$ & $\e^{7it}$ \\
\hline
$y_0$ & $A$ &0 &0 & 0\\ 
\hline
$y_1$ & $\frac{At}{2}(1-C)$ & $\frac{\I A^3}{8}$ & 0 &  0\\ 
\hline
$y_2$ & 
$\begin{matrix} 
\frac{At}{16}(-2\I + 8\I C -7\I C^2 \\ + 2t -8C t+6C^2 t)
\end{matrix}$ & 
$\begin{matrix}-\frac{\I A^3}{64}(-2\I - \I C \\  -12 t + 12 C t)\end{matrix}$
& $-\frac{5A^5}{192}$ & 0\\ \hline
$y_3$ & 
$\begin{matrix} -\frac{At}{384}(96C-210C^2+111C^3
\\
+24\I t -216\I C t + 444 \I C^2 t
\\ -252 \I C^3 t
-8t^2+104C t^2 \\
-216 C^2t^2+120C^3 t^2
 \end{matrix}$
& $\begin{matrix}
\frac{A^3}{512}(4\I - 42 \I C \\+29 \I C^2 
-92Ct \\+104C^2t
+72\I t^2 \\ -192 \I C t^2 + 120 \I C^2 t^2)
\end{matrix}$
& $\begin{matrix}  
\frac{5A^5}{4608}(-14\I-3\I C \\-60 t + 60 Ct)
 \end{matrix}$ & $-\frac{7 \I A^7}{1152}$ \\ 
\hline
&$\vdots$ &$\vdots$ &$\vdots$ &$\vdots$ \\
& $P_1(\varepsilon,t,A,B)$ & $P_3(\varepsilon,t,A,B)$
& $P_5(\varepsilon,t,A,B)$ & $P_7(\varepsilon,t,A,B)$
\end{tabular}}
\end{table}

For instance, we have
\begin{equation}\label{p5}
P_5(\varepsilon,t,A,B)
= -\frac{5A^5\varepsilon^2}{192}
+ \frac{5A^5\varepsilon^3}{4608}(-14\I-3\I C -60 t + 60 Ct) + \mathcal{O}(\varepsilon^4).
\end{equation}
In this example, $P_n$ with negative $n$ is obtained by  
$P_{-n}(\epsilon,t,A,B) = P_{n}(\epsilon,t,B,A)\left|_{\I \rightarrow -\I}\right.$.
Such a relation is not valid in general for 
potentials that are not $\e^{\I t} \leftrightarrow \e^{-\I t}$ symmetric like
$V = y \e^{3\I t} + 2\frac{dy}{dt}\e^{-\I t}$.  
\end{example}

As this example demonstrates, the secular coefficients 
are formal power series in $\varepsilon$ 
whose leading order behaves as 
\begin{align}\label{ord}
P_n(\varepsilon,t,A,B)  \sim \mathcal{O}(\varepsilon^{d_n}),\quad d_n 
\rightarrow \infty \;\,\text{as}\; \,|n| \rightarrow \infty.
\end{align}

\section{Properties of secular coefficients}\label{sec:ps}

\begin{lemma}\label{le:mkawa}
Let $s$ be an arbitrary parameter.
The formal series $y=y(\varepsilon, t)$ of the form in Eq.~(\ref{ysol}) 
that satisfies the differential equation in Eq.~(\ref{feq}) and the conditions 
\begin{align}\label{erk}
\mathrm{(i)}\;\;y(0,t) &= A \e^{\I t} + B \e^{-\I t},\qquad 
\mathrm{(ii)}\;\;
\bigl[y(\varepsilon, t)\bigr]_{\e^{\pm \I t}}\left|_{t=s}\right. = 
P_{\pm 1}(\varepsilon, s, A,B)
\end{align}
is unique and given by $y(\varepsilon, t) = Y(\varepsilon,t,A,B)$ in Eq.~(\ref{yp}).
(The left-hand side of (ii) 
means the value of the resonant secular coefficients evaluated at $t=s$.)
\end{lemma}
\begin{proof}
From the construction in the previous section, we know that the solution exists uniquely 
by choosing the parameters $\alpha_k, \beta_k$ appropriately so as to fit (ii) in 
each order of $\varepsilon$.
It is obvious that a solution $y(\varepsilon, t) = Y(\varepsilon,t,A,B)$ 
fulfills both (i) and (ii).
\end{proof}

The next lemma is the point where the specific form 
in Eq.~(\ref{V}) of the potential  $V$ matters.
\begin{lemma}\label{le:matsuko}
For arbitrary parameters $s, C$, and $D$, the formal series 
\begin{align}
\sum_{n \in \Z} P_n(\varepsilon, t-s, C,D) \e^{n \I t}
\end{align}
is also a solution to the differential equation in Eq.~(\ref{feq}).
\end{lemma}
\begin{proof}
The only nontrivial claim is that shifting $t$ to $t-s$ in $P_n$
without changing $\e^{n\I t}$ keeps it a solution. 
To show this, regard Eq.~(\ref{yp}) as a formal Laurent series in $\e^{\I t}$.
Then the original equation in Eq.~(\ref{feq}) is equivalent to the  
family of equations for the coefficient of each harmonic $\e^{m\I t}$:
for all $m \in \Z$,
\begin{equation}\label{anne}
\begin{split}
&\frac{\partial^2 P_m(\varepsilon,t,A,B)}{\partial t^2} + 
2im \frac{\partial P_m(\varepsilon,t,A,B)}{\partial t} + (1-m^2)P_m(\varepsilon,t,A,B) 
\\
&= \left[
\varepsilon V\Bigl(\varepsilon , \e^{\I t}, \e^{-\I t},
\sum_{n \in \Z} P_n(\varepsilon, t, A,B) \e^{n \I t}, 
\sum_{n \in \Z} 
\Bigl(\frac{\partial P_n(\varepsilon, t, A,B)}{\partial t} + in \Bigr) \e^{n \I t}
\Bigr)
\right]_{\e^{m \I t}}.
\end{split}
\end{equation}
In general, the right-hand side contains infinite sums like 
$\sum_{n_1+n_2+n_3=m+5} P_{n_1}P_{n_2}\frac{\partial P_{n_3}}{\partial t}$ 
when $V = \e^{-5\I t}y^2\frac{dy}{dt}$, for example.
However, thanks to Eq.~(\ref{ord}), they are actually convergent  and make sense 
as formal power series in $\varepsilon$.
The point here is that Eq.~(\ref{anne}) is a totally autonomous equation, i.e.
all the $t$ dependence is via $P_n$ and its derivatives.
Therefore, the shifted equation $t \rightarrow t-s$ is equally valid.
\end{proof}

The crux of the above argument is that Eq.~(\ref{anne}) 
is well defined and autonomous.
It clarifies why the nonautonomous part of $V$ 
has to be Laurent polynomials of $\e^{\I t}$ as in Eq.~(\ref{V}). 
Non-$\e^{\Z \I t}$-type dependence like 
$V = (2t\e^{\I t}+3t^2) y^2 \frac{dy}{dt}$ spoils the 
autonomous nature of the right-hand side of Eq.~(\ref{anne}), and hence 
invalidates the proof and the statement.
We leave consideration of a wider class of potentials like 
$V = \e^{y+y^2-\I t}$ as a future problem.

The main result of this paper is the following theorem and its consequences.
\begin{theorem}\label{th:tatsuki}
For any $s,t, A$, and $B$, the following identity between the formal series is valid:
\begin{align}
\sum_{n \in Z}P_n(\varepsilon, t,A,B)\e^{n\I t}
= \sum_{n \in \Z} 
P_n(\varepsilon,t-s,P_1(\varepsilon,s,A,B),P_{-1}(\varepsilon,s,A,B))\e^{n \I t}.
\label{cdk}
\end{align}
\end{theorem}
\begin{proof}
From Lemma \ref{le:mkawa}, 
it suffices to verify that the right-hand side of Eq.~(\ref{cdk}) is a solution 
to Eq.~(\ref{feq}) and satisfies conditions (i) and (ii) in Eq.~(\ref{erk}).
The fact that it is a solution is assured by Lemma \ref{le:matsuko}.
Checking Eq.~(\ref{erk}) is straightforward by using Eq.~(\ref{pe0}).
\end{proof}

Theorem \ref{th:tatsuki} tells us that the right-hand side 
of (\ref{cdk}) is independent of $s$.
In the RG context, it implies the independence of the choice of the initial time 
by a suitable renormalization of the amplitudes.
The novelty of Eq.~(\ref{cdk}) is that the required normalization is exactly achieved by
$P_{\pm 1}$ itself.

\begin{corollary}\label{co:sasaki}
The secular coefficients satisfy the functional relation
\begin{align}\label{toko}
P_n(\varepsilon,t,A,B) = P_n(\varepsilon,t-s, P_1(\varepsilon,s,A,B),P_{-1}(\varepsilon,s,A,B)) \quad (\forall n \in \Z).
\end{align}
\end{corollary}

Further specializing Eq.~$(\ref{toko})|_{n = \pm 1}$ to $t=0$ and 
applying Eq.~(\ref{pt0}) 
(with the reset $s \rightarrow -t$), we obtain an ``inversion" formula:
\begin{align}\label{haruko}
P_{\pm1}(\varepsilon,t,P_1(\varepsilon,-t,A,B),P_{-1}(\varepsilon,-t,A,B))=\begin{pmatrix} A\\ B\end{pmatrix}.
\end{align}

Although our derivation of Theorem \ref{th:tatsuki} and Corollary \ref{co:sasaki} 
has been quite elementary,  
their consequences are rather nontrivial.
For instance,  the first nontrivial assertion of Eq.~(\ref{toko}) 
about Eq.~(\ref{p5}) is
\begin{align}
P_5(\varepsilon,t,A,B) \equiv 
P_5\bigl(\varepsilon, t-s, A +\frac{1}{2}\varepsilon A s(1-C),  
B+\frac{1}{2}\varepsilon B s(1-C)\bigr)
\mod \mathcal{O}(\varepsilon^4).
\end{align}

The relation in Eq.~(\ref{toko}) 
is different from \cite[Eq.~(1.11)]{T} although they are 
formally similar.
In fact, the identity in Eq.~(\ref{cdk}) in terms of 
$Y(\varepsilon,t,A,B)$ of Eq.~(\ref{yp}) is {\em not} equivalent to 
$Y(\varepsilon,t,A,B) 
= Y(\varepsilon,t-s, P_1(\varepsilon,s,A,B), P_{-1}(\varepsilon,s,A,B))$
as we have $\e^{n\I t}$ rather than $\e^{n\I(t-s)}$ in the right-hand side.

\section{Renormalized amplitude and RG equation}\label{sec:ra}

We introduce $A_r(t)$ and $B_r(t)$ by either one of the 
following two sets of relations:
\begin{align}\label{sasaki}
\begin{pmatrix}
A_r(t)
\\
B_r(t) 
\end{pmatrix}
=
P_{\pm 1}(\varepsilon, t, A, B)
\;\;
\longleftrightarrow\;\;
\begin{pmatrix} A \\ B \end{pmatrix}
= P_{\pm 1}(\varepsilon, -t, A_r(t), B_r(t)).
\end{align}
Their equivalence is assured by the inversion relation in Eq.~(\ref{haruko}).
Now, Corollary \ref{co:sasaki} is stated as the identity
\begin{align}\label{deema0}
P_n(\varepsilon,t,A,B) = P_n(\varepsilon,t-s, A_r(s), B_r(s)). 
\end{align}
In particular the case $s=t$ reads 
\begin{align}\label{deema1}
P_n(\varepsilon,t, A,B) = P_n(\varepsilon,0,A_r(t), B_r(t)).
\end{align}
This relation proves that the secular $t$ dependence 
in the left-hand side can be eliminated totally by 
switching from the bare amplitudes $A,B$ to the new ones $A_r(t), B_r(t)$.
In this sense the variables $A_r(t), B_r(t)$ are called the {\em renormalized amplitudes}
\cite{CGO,NO,K}.
They allow us to rewrite the naive perturbation series in Eq.~(\ref{yp}) as
\begin{equation}\label{deema2}
\begin{split}
Y(\varepsilon,t,A,B) &= \sum_{n \in \Z}P_n(\varepsilon,0,A_r(t), B_r(t))\e^{n\I t}
\\
&= A_r(t)\e^{\I t}+B_r(t)\e^{-\I t} +  \sum_{n \in \Z\setminus \{\pm 1\}}
P_n(\varepsilon,0,A_r(t), B_r(t))\e^{n\I t}.
\end{split}
\end{equation}
By construction, 
the RG series in Eq.~(\ref{deema2}) is free of secular terms 
to all orders of $\varepsilon$.

The remaining task is to describe the dynamics or ``modulation" of the renormalized amplitudes
$A_r(t), B_r(t)$ entering Eq.~(\ref{deema2}).
The left relation in Eq.~(\ref{sasaki}) is certainly an answer, 
but there is no point in substituting it  
into Eq.~(\ref{deema2}) since it just brings us back to 
the original expansion in Eq.~(\ref{yp}), which is full of secular terms.
So we need to devise an alternative maneuver 
which suppresses the secular (nonautonomous) 
$t$ dependence totally so that $t$ always ``remains locked down in 
$A_r(t), B_r(t)$".\footnote{At the time of writing,
the number infected with COVID-19 in the world is 32356828.}
Now, with the exact renormalization in Eq.~(\ref{deema1}) at hand,
this can be done in a single line:
\begin{align}\label{aimi}
\frac{d}{dt}\!\!
\begin{pmatrix}
A_r(t) \\ B_r(t) 
\end{pmatrix}\overset{(\ref{sasaki})}{=} \!\frac{\partial P_{\pm 1}}{\partial t}(\varepsilon, t, A, B)
\overset{(\ref{deema0})}{=}
\frac{\partial P_{\pm 1}}{\partial t}(\varepsilon, t-s, A_r(s), B_r(s))
\overset{s\rightarrow t}{=}
\frac{\partial P_{\pm 1}}{\partial t}(\varepsilon, 0, A_r(t), B_r(t)).
\end{align}
In the last step we have changed $s$ to $t$.
This is allowed by the $s$ independence due to 
$\frac{\partial}{\partial s}\bigl(\frac{\partial P_\pm}{\partial t}
(\varepsilon,t-s,A_r(s),B_r(s)\bigr)
= \frac{\partial}{\partial t}\bigl(\frac{\partial P_\pm}{\partial s}
(\varepsilon,t-s,A_r(s),B_r(s)\bigr)
\overset{(\ref{deema0})}{=}0$.
From this maneuver it is clear that the $t$-derivative
in the last expression of Eq.~(\ref{aimi}) does not touch $ A_r(t), B_r(t)$.
The differential equation in Eq.~(\ref{aimi}) is called the RG or amplitude equation.
We see that the dynamics of the renormalized amplitude is 
governed by the resonant secular coefficients $P_{\pm 1}$ to all orders of $\varepsilon$.

Let us isolate the top term of the power series 
$P_{\pm 1}(\varepsilon,t,A,B)$ and name the other part  as $Q_{\pm 1}(\varepsilon, t, A, B)$:
\begin{align}
P_{\pm 1}(\varepsilon,t,A,B) = \begin{pmatrix}A  \\  B \end{pmatrix} 
+ \varepsilon Q_{\pm 1}(\varepsilon, t, A, B).
\label{qdef}
\end{align}
Then $Q_{\pm 1}(\varepsilon,t,A,B)= \sum_{k \ge 1}\varepsilon^{k-1} f_{\pm 1,k}(t)$ 
is still a power series in $\varepsilon$ such that 
\begin{align}
Q_{\pm 1}(\varepsilon,0,A,B) = 0
\label{qt0}
\end{align}
because of Eq.~(\ref{pt0}).
Now the RG equation in Eq.~(\ref{aimi}) is simplified slightly  as
\begin{align}\label{rg}
\frac{d}{dt}
\begin{pmatrix}
A_r(t) \\ B_r(t) 
\end{pmatrix} = \varepsilon
\frac{\partial Q_{\pm 1}}{\partial t}(\varepsilon, 0, A_r(t), B_r(t)),
\end{align}
where, as in Eq.~(\ref{aimi}),  
the $t$-derivative in the right-hand side does not concern $A_r(t), B_r(t)$.
This representation indicates that the RG dynamics is certainly ``slow" 
in the sense that the right-hand side is 
at last of order $\mathcal{O}(\varepsilon)$.

\vspace{0.3cm}
In the earlier works \cite{K,NO}, the right relation in 
Eq.~(\ref{sasaki}) was conventionally formulated as
\begin{equation}
A = A_r(t)Z_a(\varepsilon, t,A_r(t),B_r(t)), 
\qquad B = B_r(t)Z_b(\varepsilon,t,A_r(t),B_r(t))
\end{equation} 
by further introducing the so-called the renormalization constants $Z_a, Z_b$.
Moreover, reversing these relations had to be  attributed to the implicit function theorem.
One of the main achievements in this paper is the manifest bijection 
in Eq.~(\ref{sasaki}) 
between the bare and the renormalized amplitudes
that untangles these issues, 
and having elucidated its elegant origin in the functional relation of (\ref{toko}).
The abovementioned renormalization constants, 
although they can now be dispensed with, acquire a ``closed formula" as
\begin{align}
Z_a(\varepsilon, t,A_r(t),B_r(t)) & = \frac{P_1(\varepsilon, -t,A_r(t),B_r(t))}{A_r(t)}
= 1 + \varepsilon  \frac{Q_1(\varepsilon, -t,A_r(t),B_r(t))}{A_r(t)},
\\
Z_b(\varepsilon, t,A_r(t),B_r(t)) &= \frac{P_{-1}(\varepsilon, -t,A_r(t),B_r(t))}{B_r(t)}
= 1 + \varepsilon\frac{Q_{-1}(\varepsilon, -t,A_r(t),B_r(t))}{B_r(t)}.
\end{align}

In general, the solution $y(t)$ can either be stable 
around the nonperturbative one $y_0(t)$,  (\ref{y0}),
or unstably growing depending on $V$ no matter how small $\varepsilon$ is.
On general grounds we expect that the bare amplitudes $A$ and $B$ should be 
small enough for stability in the long time scale.
A quantitative result on such issues in a similar system 
is available, for example in \cite[Theorem 2.7]{C}.

\section{Examples}\label{sec:ex}

\subsection{Van der Pol equation}
We consider Eq.~(\ref{vdp}), which was also discussed in Example \ref{ex:vdp}.
Introduce the variables $R = R(t)$ and $\theta = \theta(t)$ connected 
to the renormalized amplitudes\footnote{This change of variables is optional.
The original $A_r(t)$ and $B_r(t)$ equally suit the numerical work.
The same feature applies to the other equations in this section.}
\begin{align}\label{Rte1}
A_r(t) =  R(t)\e^{\I \theta(t)}, \quad 
B_r(t) = R(t)\e^{-\I \theta(t)}.
\end{align}
Set $\tau = t + \theta(t)$.
Then the renormalized expansion, Eq.~(\ref{deema2}), reads
\begin{equation}
\begin{split}
y &= 2R \cos \tau -\frac{\varepsilon R^3}{4}\sin3\tau 
-\frac{\varepsilon^2R^3}{96}\bigl(6 \cos3 \tau + R^2 (3 \cos3 \tau + 5 \cos5 \tau)\bigr)
\\
&-\frac{\varepsilon^3R^3}{2304}\Bigl(
36 \sin3 \tau - 14 R^2 (27 \sin3 \tau + 5 \sin5 \tau) 
\\
& \qquad \qquad +  R^4 (261 \sin3 \tau - 15 \sin5 \tau - 28 \sin7 \tau)
\Bigr)+ \mathcal{O}(\varepsilon^4).
\end{split}
\end{equation}
The RG equation, Eq.~(\ref{rg}), is given by
\begin{align}
\frac{d \log R}{dt}&=\frac{\varepsilon(1-R^2)}{2}
-\frac{\varepsilon^3 R^2(32-70R^2+37R^4)}{128} \nonumber\\
& + \frac{\varepsilon^5R^4(-1980+8154R^2-10757R^4+4589R^6)}{36864}
\nonumber \\
&- \frac{\varepsilon^7R^4}{21233664}(2950992 - 16173432 R^2 + 28047688 R^4 \nonumber \\
& \qquad \qquad  \qquad - 14916436 R^6 - 4396557 R^8 + 
 4493323 R^{10})
 + \mathcal{O}(\varepsilon^9),
\\
\frac{d\theta}{dt} &= 
\frac{\varepsilon^2(-2+8R^2-7R^4)}{16}
+\frac{\varepsilon^4(-24-192R^2+1020R^4-1266R^6+497R^8)}{3072}
\nonumber\\
&+\frac{\varepsilon^6}{1769472}
\bigl(-1728-6912R^2+181872R^4-455608R^6  \nonumber \\
&\qquad \qquad  \qquad +121432R^8+417540R^{10}-266949R^{12}\bigr)
+ \mathcal{O}(\varepsilon^8).
\end{align}
By postulating $\frac{d \log R}{dt}=0$, we can find the values on the limit cycle:
\begin{align}
2R_c &=  2+ \frac{\varepsilon^2}{64}-\frac{23\varepsilon^4}{49152} - \frac{51619\varepsilon^6}{169869312}
+ \mathcal{O}(\varepsilon^{8}),
\label{rc}
\\
\bigl(\frac{d\theta}{dt}\bigr)_c &= -\frac{\varepsilon^2}{16}
+ \frac{17 \varepsilon^4}{3072}
+ \frac{35\varepsilon^6}{884736} + \mathcal{O}(\varepsilon^8).
\end{align}
The approximate leading value $2R_c=2$ is well known from the energy balance argument
that the total work by the friction term during a period should be zero, i.e.   by requiring 
$\int_0^{2\pi}(y^2-1)\bigl(\frac{dy}{dt}\bigr)^2 dt= 0$ for $y=2R_c\cos t$.

\subsection{Mathieu equation}
We consider Eq.~(\ref{mathi}) with $g$ dependent on $\varepsilon$ as 
$g= g_1 +  g_2\varepsilon +  g_3\varepsilon^2  + \cdots$.
Then, $Q_1(\varepsilon,t ,A,B) $ defined  in Eq.~(\ref{qdef}) is given by
\begin{equation}
\begin{split}
Q_1(\varepsilon,t ,A,B) 
&= \frac{\I A g_1 t}{2} -\frac{\varepsilon t}{24} \bigl(\I (8 A + 12 B + 3 A g_1^2 - 12 A g_2) + 3 A g_1^2 t \bigr)
\\
&+\frac{\varepsilon^3 t }{144}
\bigl(\I (88 A g_1 + 72 B g_1 + 9 A g_1^3 - 36 A g_1 g_2 + 72 A g_3) \\
&\qquad \qquad +  3 A g_1 (8 + 3 g_1^2 - 12 g_2) t - 3 \I A g_1^3 t^2
\bigr)+ \mathcal{O}(\varepsilon^4).
\end{split}
\end{equation}
The other one is obtained by
$Q_{-1}(\varepsilon,t ,A,B)  =  Q_1(\varepsilon,t ,B,A)\vert_{\I \rightarrow -\I}$. 
This example is exceptional among those in 
Eqs.~(\ref{vdp})--(\ref{ray}) in that it is the only equation
which is linear and,  moreover,  nonautonomous.
Reflecting the former feature, the RG equation also becomes linear.
In fact,  differentiation of Eq.~(\ref{rg}) 
can be combined and split into the two identical equations
\begin{align}\label{mab}
\frac{d^2A_r(t)}{dt^2} = -\omega^2 A_r(t), \qquad 
\frac{d^2B_r(t)}{dt^2} = -\omega^2 B_r(t),
\end{align}
where the constant $\omega^2$ is given by 
\begin{equation}
\begin{split}
\omega^2 = &\frac{\varepsilon^2}{4}  g_1^2 - \frac{\varepsilon^3}{24}  g_1 (8 + 3 g_1^2 - 12 g_2) 
\\ &-\frac{\varepsilon^4 }{576} (80 - 400 g_1^2 - 45 g_1^4 + 192 g_2 + 216 g_1^2 g_2 - 
    144 g_2^2 - 288 g_1 g_3)
    \\
&-\frac{\varepsilon^5}{3456} (-840 g_1 + 3920 g_1^3 + 189 g_1^5 - 4800 g_1 g_2 - 1080 g_1^3 g_2 + 
   1296 g_1 g_2^2 + 1152 g_3 \\
   &\qquad \qquad   + 1296 g_1^2 g_3 - 1728 g_2 g_3 - 
   1728 g_1 g_4) + \mathcal{O}(\varepsilon^6).
\end{split}
\end{equation}
The stable region $\omega^2>0$ and the unstable region $\omega^2<0$ 
of $A_r(t), B_r(t)$ are separated by the curve $\omega^2=0$.
Solving it order by order in $\varepsilon$ with respect to $g_1, g_2, g_3, \ldots$ 
yields two branches.
Let us present them for the combination $a:= 1 + \varepsilon g$
which is the usual coupling constant in the conventional setting of 
the Mathieu equation in Eq.~(\ref{mathi}) as 
\begin{align}\label{czk}
\frac{d^2y}{dt^2} + (a+ 2\varepsilon \cos t) y=0.
\end{align}
Then the branches are $a=a^{\pm}$, where
\begin{align}\label{apm}
a^- &= 1 -\frac{\varepsilon^2}{3} + \frac{5\varepsilon^4}{216} + \cdots, \qquad 
a^+ = 1 +\frac{5\varepsilon^2}{3} - \frac{763\varepsilon^4}{216} + \cdots.
 \end{align}
These curves in the $(a, \varepsilon)$ plane 
specify the boundaries of the unstable region $a^- <  a < a^+$
and the stable region in the vicinity of $(a, \varepsilon) = (1,0)$, 
reproducing part of Ref.~\cite[Fig.11.11]{BO}.
The separation into stable and unstable regions is a manifestation of the 
band structure in the Floquet-Bloch theory in the context of the 
Schr\"odinger equation in one-dimensional 
periodic potentials.

The result $a^\pm$ in Eq.~(\ref{apm}) agrees with the zeros of the determinants 
$\Delta^{\pm}(\varepsilon,a)=0$ of the 
semi-infinite Jacobi matrices\footnote{Symmetric, tridiagonal matrices with positive 
off-diagonal elements.  We imagine $\varepsilon$ is positive to reply on this nomenclature.
$\Delta^-(\varepsilon,a)$ and $\Delta^+(\varepsilon,a)$ 
correspond to $S_e(x)$ and $C_e(x)$ in Ref.~\cite[p.~176]{I}, respectively.}
 around the resonance \cite[sec.7-1]{I}, where 
\begin{align*}
\Delta^-(\varepsilon,a) &= \begin{vmatrix}
a-1^2 & \varepsilon &  & & & \\
\varepsilon & a-2^2 & \varepsilon &  & & \\
& \varepsilon & a-3^2 & \varepsilon & & \\
& & \varepsilon & a-4^2 & \varepsilon &  \\
 &  & &  \varepsilon &  \ddots &  
\end{vmatrix},
\\
\Delta^+(\varepsilon,a) &= \begin{vmatrix}
\frac{a}{2} & \varepsilon &  & & \\
\varepsilon & a-1^2 & \varepsilon   & & & \\
& \varepsilon & a-2^2 & \varepsilon &   & \\
& & \varepsilon & a-3^2 & \varepsilon  & \\
 &  & & \varepsilon &  \ddots & 
\end{vmatrix}.
\end{align*}

It is known that Mathieu equation in Eq.~(\ref{czk}) 
is in parametric resonance at infinitely many 
points $a= \frac{n^2}{4}  \,(n=1,2,3,\ldots)$ \cite[Sect.~11.4]{BO}.
The equation in Eq.~(\ref{mathi}) and the result in Eq.~(\ref{apm}) 
correspond to the $n=2$ resonance.
At $a= \frac{n^2}{4}$,  Eq.~(\ref{czk}) becomes  
$\frac{d^2y}{ds^2} + y + \frac{8\varepsilon}{n^2}y\cos(\frac{2s}{n})=0$
by switching to a new time variable $s=\frac{nt}{2}$.
Thus, a similar analysis to this paper is also possible for the $n=1$  resonance.
On the other hand the region around  $a= \frac{n^2}{4}$ with $n\ge 3$ 
is outside our assumption on $V$.
It will be an interesting exercise to see how the RG series 
fits the exact solution \cite{D}.

\subsection{Duffing equation}

Consider Eq.~(\ref{duff}) with $g=1$, 
which can be attained by $y\rightarrow y/\sqrt{g}$.
We have included a $\frac{dy}{dt}$ term since otherwise 
the equation is integrable by an elliptic function.
Introduce the variables $R = R(t)$ and $\theta = \theta(t)$ 
connected to the renormalized amplitudes by Eq.~(\ref{Rte1})
and set $\tau = t + \theta(t)$.
Then the renormalized expansion in Eq.~(\ref{deema2}) reads
\begin{equation}
\begin{split}
y &= 2R \cos \tau + \frac{\varepsilon R^3}{4} \cos 3\tau
+ \frac{\varepsilon^2 R^3}{32}\bigl(6 \sin 3\tau+R^2(\cos 5\tau-21\cos 3\tau)\bigr)\\
 &+ \frac{\varepsilon^3 R^3}{768}
 \Bigl(
-36\cos 3\tau -2R^2(567\sin 3\tau - 19 \sin 5\tau) \\
&\qquad \qquad +3R^4(417\cos3\tau-43\cos 5\tau + \cos 7\tau)
 \Bigr)
 + \mathcal{O}(\varepsilon^4).
 \end{split}
 \end{equation}
The RG equation in Eq.~(\ref{rg}) is given by
\begin{align}
\frac{d \log R}{dt}&=
 -\frac{\varepsilon}{2}
 +\frac{3\varepsilon^2  R^2}{4} 
-\frac{195 \varepsilon^3  R^4}{64}
+ \frac{5931\varepsilon^4  R^6}{512} + 
\frac{\varepsilon^5R^4(16092-172027R^4)}{4096}
+ \mathcal{O}(\varepsilon^6),
\\
\frac{d\theta}{dt} & = \frac{3\varepsilon R^2}{2}
-\frac{\varepsilon^2(2+15R^4)}{16}
-\frac{3\varepsilon^3 R^2(8-41R^4)}{128}
+\frac{\varepsilon^4(-8+4116 R^4-921R^8)}{1024}
\nonumber\\
&-\frac{3\varepsilon^5R^2(8+21305R^4-193R^8)}{2048}
+ \mathcal{O}(\varepsilon^6).
\end{align}

\subsection{Rayleigh equation}

We consider Eq.~(\ref{ray}).
Introduce the variables $R = R(t)$ and $\theta = \theta(t)$ connected to the renormalized amplitudes by Eq.~(\ref{Rte1})
and set $\tau = t + \theta(t)$.
Then the renormalized expansion in Eq.~(\ref{deema2}) reads
\begin{equation}
\begin{split}
y &= 2R \cos \tau + \frac{\varepsilon R^3}{12}\sin 3\tau +
\frac{\varepsilon^2 R^3}{96}\bigl(-6\cos 3\tau +R^2(9\cos 3\tau -\cos 5\tau)
\bigr)
\\
&+ \frac{\varepsilon^3 R^3}{2304}\Bigl(
-36 \sin3 \tau - 2 R^2 (63 \sin3 \tau + 17 \sin5 \tau) 
\\
&\qquad \qquad  +  R^4 (111 \sin3 \tau + 51 \sin5 \tau - 4 \sin7 \tau)
\Bigr)+ \mathcal{O}(\varepsilon^4).
\end{split}
\end{equation}
The RG equation in Eq.~(\ref{rg}) is given by
\begin{align}
\frac{d \log R}{dt}&=
\frac{\varepsilon(1-R^2)}{2}+\frac{\varepsilon^3R^4(22-13R^2)}{128}
\nonumber \\
&-\frac{\varepsilon^5 R^4(2268-1026R^2-2683R^4+1603R^6)}{36864}+ \mathcal{O}(\varepsilon^7),
\\
\frac{d\theta}{dt} & =
\frac{\varepsilon^2(R^4-2)}{16}
+ \frac{\varepsilon^4(-24+156R^4-234R^6+65R^8)}{3072}
\nonumber\\
&+ \frac{\varepsilon^6(-1728-98064R^4+305208R^6-210728R^8-71388R^{10}+84627R^{12})}{1769472}
+ \mathcal{O}(\varepsilon^8).
\end{align}

\subsection{A nonlinear and nonautonomous example}\label{ss:saori}

Finally, we consider a nonlinear and nonautonomous example:
\begin{align}\label{hiro}
\frac{d^2y}{dt^2} + y = 2\varepsilon \frac{dy}{dt} y\cos t.
\end{align}
Introduce the variables $R = R(t)$ and $\theta = \theta(t)$ 
connected to the renormalized amplitudes by Eq.~(\ref{Rte1}).
Then the renormalized expansion in Eq.~\ref{deema2}) reads
\begin{equation}\label{gs}
\begin{split}
y&= 2R \cos(\theta+t) + \frac{1}{4}\varepsilon R^2\sin(2\theta+3t)-
\frac{\varepsilon^2 R^3}{24}
\bigl(3 \cos \theta  \cos (2 \theta +3 t)+\cos (3 \theta +5 t) \bigr)
\\
&+\frac{\varepsilon^3R^4}{4608} \bigl(
-24 \sin (2 \theta +3 t)-33 \sin (4 \theta +7 t)+14 \sin \theta  \cos (3 \theta +5 t)+72
   \sin 2 \theta \cos (2 \theta +3 t)
   \\
   &+288 \cos 2 \theta  \sin (2 \theta +3 t)-146 \cos \theta \sin (3 \theta +5 t)
\bigr)+ \mathcal{O}(\varepsilon^4).
\end{split}
\end{equation}
The RG equation in Eq.~(\ref{rg}) is given by
\begin{align}
\frac{d\log R}{dt} &=
\frac{1}{2} \varepsilon R \cos \theta-\frac{1}{4} \varepsilon^2R^2 \sin 2\theta
+\frac{5}{16} \varepsilon^3 R^3 \cos \theta
\nonumber  \\
&-\varepsilon^4R^4 
   \left(\frac{21}{64} \sin 2\theta+\frac{1}{32}  \sin 4\theta\right)+O\left(\varepsilon^5\right),
 \label{rag1}\\
 \frac{d\theta}{dt} &=
 \frac{1}{2} \varepsilon R \sin \theta
-\varepsilon^2R^2 \left(\frac{1}{4}  \cos 2\theta+\frac{3 }{8}\right)+\frac{1}{8} \varepsilon^3 R^3 \sin \theta 
\nonumber \\
&+\frac{1}{128} \varepsilon^4R^4 \left(9\cos 2\theta-4 \cos 4 \theta -3\right)+O\left(\varepsilon^5\right).
 \label{rag2}
\end{align}
Unlike the Van de Pol, Duffing and Rayleigh 
equations, we have the essential mixture of $R$ and $\theta$
in the right-hand side of the RG equation reflecting 
the nonautonomous and nonlinear nature of Eq.~(\ref{hiro}).
As shown in Fig.~1, 
there seems only one peak in the envelop $R(t)$ in a certain parameter range.
 
\begin{figure}[H]
\begin{center}
\includegraphics[scale=0.55]{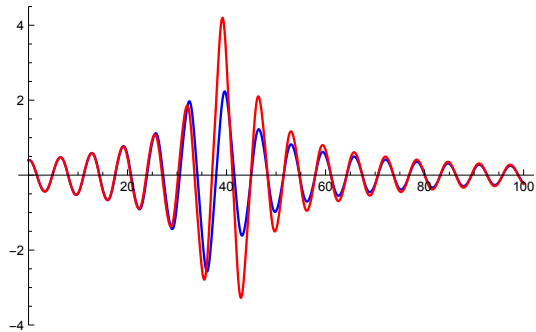}
\caption{Plot of $y(t)$ vs $t$ by the direct numerical solution of 
Eq.~(\ref{hiro}) (blue) and 
the RG expansion (red) started from the same initial condition 
$R(0)=0.2, \theta(0)=-0.1$ with $\varepsilon=0.25$.
We have kept only the $\varepsilon^0, \varepsilon^1$ terms in Eq.~(\ref{gs}) 
and the leading 
$\varepsilon^1$ term in Eqs.~(\ref{rag1}) and (\ref{rag2}).
Taking the $\varepsilon^2$ term in Eqs~(\ref{rag1}) and  (\ref{rag2}) 
into account already makes it too difficult to observe the discrepancy.}
\end{center}
\label{fig}
\end{figure}

\section*{Acknowledgments}

The author thanks Yoshitsugu Oono for a communication around 2007.
He is also grateful to the anonymous referee for productive comments.
This work is supported by 
Grants-in-Aid for Scientific Research Nos.~16H03922, 18H01141 and
18K03452 from JSPS.


%

\vspace{0.2cm}
\noindent


\let\doi\relax

\end{document}